\documentclass[final]{dmtcs-episciences}

\usepackage{amsmath}
\usepackage[utf8]{inputenc}
\usepackage[T1]{fontenc}
\usepackage{amsfonts}
\usepackage{amssymb}
\usepackage{float}
\usepackage{url}
 \usepackage[all]{xy}
 \usepackage{textgreek}
 \usepackage{dsfont}
 \usepackage{stmaryrd}
\usepackage{pxfonts}
\usepackage{bbm}
\usepackage{tikz}
\usepackage[fancy]{tikz-inet}
\usetikzlibrary{automata}
\usetikzlibrary{arrows}
\usetikzlibrary{shapes}
\usetikzlibrary{decorations.pathmorphing}
\usetikzlibrary{fit}
\usepackage{ifthen}
\usepackage{stmaryrd} 

\tikzstyle{every picture}=[
  >=stealth', 
shorten >=1pt, 
node distance=1.44cm,
auto,
bend angle=45,
initial text=,
  every state/.style={inner sep=0.75mm, minimum size=1mm},
font=\small,
]

\def\sc{\mathrm{sc}}

\newcommand{\IntEnt}[1]{\llbracket #1\rrbracket}

\newtheorem{definition}{Definition}
\newtheorem{lemma}{Lemma}
\newtheorem{example}{Example}
\newtheorem{proposition}{Proposition}
\newtheorem{theorem}{Theorem}
\newtheorem{claim}{Claim}
\newtheorem{rmk}{Remark}
\newtheorem{corollary}{Corollary}

\def\N{{\mathbb N}}
\begin{document}

\title{New tools for state complexity}
\author{Pascal Caron\affiliationmark{1}
    \and Edwin Hamel-De Le Court\affiliationmark{1}
    \and Jean-Gabriel Luque\affiliationmark{1}\\
    \and Bruno Patrou\affiliationmark{1}}
\affiliation{LITIS, Universit\'e de Rouen-Normandie,  France}
\maketitle
\received{2018-9-18}
\revised{2019-9-25,2020-1-7,2020-1-22}
\accepted{2020-2-12}
\publicationdetails{22}{2020}{1}{9}{4835}

\begin{abstract}
A monster is an automaton in which every function from states to states is represented by at least one letter. A modifier is a set of functions allowing one to transform a set of automata into one automaton.
We revisit some language transformation algorithms in terms of modifier and monster. These new theoretical concepts allow one to find easily some state complexities. We illustrate this by retrieving the state complexity of the Star of Intersection and the one of the Square root operation. 
\end{abstract}

\section{Introduction}
The studies around state complexities last for more than twenty years now. Mainly initiated by Yu et al (\cite{YZS94}) and very active ever since, this research area dates back in fact to the beginning of the 1970s. In particular, in \cite{Mas70} Maslov gives values (without proofs) for the state complexity of some operations: square root, cyclic shift and proportional removal. From these foundations, tens and tens of papers have been produced and different sub-domains have appeared depending on whether the used automata are deterministic or not, whether the languages are finite or infinite, belong to some classes (codes, star-free, $\ldots$) and so on. We focus here on the deterministic case for any language.

The state complexity of a regular language is the size of its minimal automaton and the state complexity of a regular operation is the maximal one of those languages obtained by applying this operation onto languages of fixed state complexities. So, to compute a state complexity, most of the time the approach is to calculate an upper bound from the characteristics of the considered operation and to provide a witness, that is a specific example reaching the bound which is then the desired state complexity.

This work has been done for numerous unary and binary operations. See, for example, \cite{Dom02}, \cite{JJS05}, \cite{Jir05}, \cite{JO08}, \cite{Yu01a} and \cite{GMRY17} for a survey of the subject. More recently, the state complexity of combinations of operations has also been studied. In most of the cases the result is not simply the mathematical composition of the individual complexities and studies lead to interesting situations. Examples can be found in \cite{SSY07}, \cite{CGKY11}, \cite{GSY08} or \cite{JO11}.

Beyond the search of state complexities and witnesses, some studies try to improve the given witnesses, especially the size of their alphabet (\cite{CLP18}, \cite{CLP17}). Others try to unify the techniques and the approaches used to solve the different encountered problems. In \cite{Brz13}, Brzozowski proposes to use some fundamental configurations to produce witnesses in many situations. In \cite{CLMP15}, the authors show how to compute the state complexities of $16$ combinations by only studying three of them.

In this paper, we propose a general method to build witnesses, consisting in maximizing the transition function of automata. Among the resulting automata, called monsters, at least one of them is a witness. We just have to discuss the finality of the states to determine which ones are. We illustrate this technique by recomputing the state complexity of the operation obtained in combining star with intersection. The state complexity of the square root operation is also computed and improved (compared to the bound given by Maslov \cite{Mas70}) as another illustration.

The paper is organized as follows. Section \ref{sect-prel}  gives definitions and notations about automata and combinatorics. In Section  \ref{sect-modifier}, we define  \emph{modifiers} and  give  some properties of these algebraic structures. In Section \ref{sect-monsters}, monsters automata are defined and  their use in automata computation is shown. Section \ref{sect-appli} is devoted to   show how these new tools can be used to compute tight bounds for state complexity. Star of intersection and square root  examples are described.
\section{Preliminaries}\label{sect-prel}
Let $\Sigma$ 
denote a finite alphabet. A word $w$ over  $\Sigma$ is a finite sequence of symbols of $\Sigma$. The length of  $w$, denoted by $|w|$, is the number of occurrences of symbols of $\Sigma$ in $w$. For  $a\in \Sigma$, we denote by $|w|_a$ the number of occurrences of  $a$ in $w$.  The set of all finite words over $\Sigma$ is denoted by $\Sigma ^*$.  The empty word is denoted by  $\varepsilon$. A language is a subset of $\Sigma^*$. The cardinality of a finite set $E$ is denoted by $\#E$,  the set of subsets of  $E$ is denoted by $2^E$ and the set of mappings of $E$ into itself is denoted by $E^E$.


A  finite automaton (FA) is a $5$-tuple $A=(\Sigma,Q,I,F,\delta)$ where $\Sigma$ is the input alphabet, $Q$ is a finite set of states, $I\subset Q$ is the set of initial states, $F\subset Q$ is the set of final states and $\delta$ is the transition function from  $Q\times \Sigma$ to $2^Q$ extended in a natural way from $2^Q\times \Sigma^*$ to $2^Q$. 

 A word $w\in \Sigma ^*$ is recognized by an FA $A$ if $\delta(I,w)\cap F\neq \emptyset$. 
The language recognized by an FA $A$ is the set $L(A)$ of words recognized by $A$. 
Two automata are said to be equivalent if they recognize the same language.  
A state $q$ is accessible in an FA  if there exists a word $w\in \Sigma ^*$ such that $q\in \delta(I,w)$.

An FA is complete and deterministic (CDFA) if $\#I=1$ and for all $q\in Q$, for all $a\in \Sigma$, $\#\delta(q,a)= 1$. 
Let $D=(\Sigma,Q_D,i_D,F_D,\delta)$ be a CDFA. When there is no ambiguity, we identify $\#D$ to $\#Q_D$. For any word $w$, we denote by $\delta^w$ the function $q\rightarrow\delta(q,w)$.
Two states $q_1,q_2$ of  $D$ are equivalent if for any word $w$ of $\Sigma^*$, $\delta(q_1, w)\in F_D$ if and only if $\delta(q_2, w)\in F_D$. Such an equivalence is denoted by $q_1\sim q_2$. A CDFA is  minimal if there does not exist any equivalent   CDFA  with less states and it is well known that for any DFA, there exists a unique minimal equivalent one \cite{HU79}. Such a minimal CDFA  can be  obtained from $D$ by computing the accessible part of the automaton $D/\sim=(\Sigma,Q_D/\sim,[i_D],F_D/\sim,\delta_{\sim})$ where for any $q\in Q_D$, $[q]$ is the $\sim$-class of the state $q$ and satisfies the property  $\delta_{\sim}([q],a)=[\delta(q,a)]$, for any $a\in \Sigma$. The number of its states is denoted by $\#_{Min}(D)$.
In a minimal CDFA, any two distinct states are pairwise inequivalent. 

For any integer $n$, let us denote $\llbracket n\rrbracket$ for $\{0,\ldots, n-1\}$. When there is no ambiguity, for any character $\mathtt{X}$ and any integer $k$ given by the context, we  write $\underline{\mathtt{X}}$ for  $(\mathtt{X}_1,\cdots, \mathtt{X}_k)$.
 The state complexity of a regular language $L$ denoted by $\sc(L)$ is the number of states of its minimal CDFA. 
  Let ${\cal L}_n$ be the set of languages of state complexity $n$. The state complexity of a unary operation $\otimes$ is the function $\sc_{\otimes}$ associating with an integer $n$, the maximum of the state complexities of $\otimes L$ for $L\in {\cal L}_n$.
  A language $L\in {\cal L}_n$ is a witness (for $\otimes$) if  $\sc(\otimes(L))=\sc_{\otimes}(n)$.
  This can be generalized, and the state complexity of a $k$-ary operation $\otimes$ is the $k$-ary function which associates with any $k$-tuple of integers $\underline{n}$, the integer $\mathrm{max}\{\sc(\otimes\underline{L})\mid \underline{L}\in\mathcal{L}_{n_1}\times \dots \times \mathcal{L}_{n_k}\}$. Then, a witness is a tuple $\underline{L}\in({\cal L}_{n_1}\times \cdots  \times{\cal L}_{n_k})$ such that $\sc\underline{L}=\sc_{\otimes}\underline{n}$. 

We  also need some background from finite transformation semigroup theory \cite{GM08}. 
Let $n$ be an integer. A transformation $t$ is an element of $\IntEnt{n}^{\IntEnt{n}}$.
We denote by $it$ the image of $i$ under $t$. A transformation of $\IntEnt{n}$ can be represented by $t=[i_0, i_1, \ldots i_{n-1}]$ which means that $i_k=kt$ for each $k\in \IntEnt{n}$ and $i_k\in \IntEnt{n}$. A \textit{permutation} is a bijective transformation on $\IntEnt{n}$. The \textit{identity} permutation is denoted by $\mathds{1}$. A \textit{cycle} of length $\ell\leq n$  is a permutation $c$, denoted   by $(i_0,i_1,\ldots, i_{\ell-1})$, on a subset $I=\{i_0,\ldots ,i_{\ell-1}\}$ of $\IntEnt{n}$  where  $i_kc=i_{k+1}$ for $0\leq k<\ell-1$ and $i_{\ell-1}c=i_0$.  
A \textit{transposition} $t=(i,j)$ is a permutation on $\IntEnt{n}$ where $it=j$ and $jt=i$ and for all  elements $k\in \IntEnt{n}\setminus \{i,j\}$, $kt=k$.  A \textit{contraction}  $t=\left(i\atop j\right)$ 
is a transformation where  $it=j$ and  for all  elements $k\in \IntEnt{n}\setminus \{i\}$, $kt=k$.

Let $L$ and $L'$ be two regular languages defined over an alphabet $\Sigma$. Let $\mathrm{Union}(L,L')=\{w\mid w\in L \vee w\in L'\}$,  $\mathrm{Inter}(L,L')=\{w\mid w\in L \wedge w\in L'\}$,  $\mathrm{Xor}(L,L')=\{w\mid (w\in L \wedge w\notin L') \vee (w\notin L \wedge w\in L')\}$,  $\mathrm{Prefin}(L)=\{w=uv\mid u\in L,\ v\in \Sigma^*\}$,  $\mathrm{Comp}(L)=\{w\mid w\not\in L\}$, $\mathrm{Conc}(L,L')=\{w=uv\mid u\in L,\ v\in L'\}$, $\mathrm{Star}(L)=\{w=u_1\cdots u_n\mid  u_i\in L\}$, $\mathrm{SRoot}(L)=\{w\in \Sigma^*\mid ww\in L\}$ .

\section{Modifier and associated transformations}\label{sect-modifier}

We first define a mechanism which unifies some automata transformations for  regular operations on languages. This mechanism is called a \emph{modifier}. A $k$-modifier is an algorithm taking $k$ automata as input and outputting an automaton. A lot of regular operations on languages can be described using this mechanism (mirror, complement, Kleene star, \ldots). These regular operations are called \emph{describable}.
Then, we give some properties for describable operations.
 We will first see that not all regular operations are describable and that there also exist modifiers which do not correspond to regular operations on languages.

\subsection{Definitions}

  The \emph{state configuration} of a DFA $A=(\Sigma,Q,i,F,\delta)$ is the triplet $(Q,i,F)$.
Our purpose is to consider operations on languages that can be encoded on DFA. To this aim, such an operation will be described as a $k$-ary operator $\mathfrak m$, acting on DFAs $A_1,\ldots A_k$ over the same alphabet $\Sigma$ and producing a new DFA such that
 \begin{itemize}
  \item the alphabet of $\mathfrak m (A_1,...,A_k)$ is $\Sigma$,
  \item the state configuration of $\mathfrak m (A_1,...,A_k)$ depends only on the state configurations of the DFAs $A_1,\ldots,A_k$,
  \item for any letter $a\in\Sigma$, the transition function of $a$ in $\mathfrak m (A_1,\ldots,A_k)$ depends only on the state configurations of the DFAs $A_1,\ldots, A_k$ and on the transition functions of $a$ in each of the DFAs $A_1,...,A_k$ (not on the letter itself nor on any other letter or transition function).
  \end{itemize}
More formally,
\begin{definition}
A $k$-modifier $\mathfrak m$ is a  $4$-tuple of mappings $(\mathfrak Q,\iota,\mathfrak f,\mathfrak d)$ acting on $k$ CDFA $\underline{A}$ with $A_j=(\Sigma, Q_j,i_j,F_j,\delta_ j)$ to build a CDFA $\mathfrak m\underline{A}=(\Sigma,Q,i,F,\delta)$, where 

\begin{center}
$Q=\mathfrak Q \underline{Q},\
i=\iota(\underline{Q},\underline{i},\underline{F}), 
F=\mathfrak f(\underline{Q},\underline{i},\underline{F})$ and \\
\smallskip
$\forall a\in \Sigma,\ \delta^a=\mathfrak d(\underline{i},\underline{F},\underline{\delta^a}).
$ 
\end{center}
\end{definition}
Notice that we do not need to put explicitly the dependency of $\mathfrak d$ on $\underline{Q}$ because the information is already present in $\underline{\delta^a}$.

 For $1$-modifiers, as $\underline{Q}=(Q_1)$, we  denote $\iota(Q_1,i_1,F_1)$ for $\iota (\underline{Q},\underline{i},\underline{F})$,  $\mathfrak f(Q_1,i_1,F_1)$ for $\mathfrak f(\underline{Q},\underline{i},\underline{F})$, and $\mathfrak d(i_1,F_1, \delta_1^a)$ for $\mathfrak d(\underline{i},\underline{F}, \underline{\delta^a})$.
\begin{example}\label{exAutMod}
	\rm Consider the modifier $\mathfrak{Prefin}$ of Table \ref{table operations}. If $A_{1}=(\Sigma,Q_{1},i_1,F_{1},\delta_1)$ is a complete deterministic automaton then ${\mathfrak{Prefin}}(A_{1})=(\Sigma,Q_{1},i_1,F_{1},\delta)$ where for any state $q\in Q_{1}$ for any $a\in \Sigma$ we have
	$\delta^a(q)=\delta_1^a(q)$ if $q\not \in F_1$ and $\delta^a(q)=q$ if $q\in F_1$ .\\
	For instance consider the automaton $A_{1}$ with the following graphical representation:
	
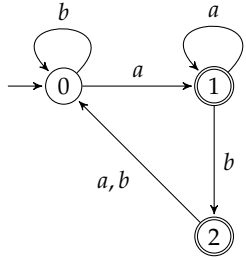
\begin{center}
\begin{tikzpicture}[node distance=2cm]
\node [state,initial] (p0) {$0$};
\node[state,accepting] (p1)[right of=p0]{$1$};
\node[state,accepting] (p2)[below of=p1]{$2$};
\path[->]
(p0)edge node{$a$} (p1)
(p0)edge[loop ] node[swap] {$b$} (p0)
(p1)edge[loop] node[swap] {$a$} (p1)
(p1)edge node {$b$} (p2)
(p2)edge node {$a,b$} (p0)
;
\end{tikzpicture}
\end{center}
The automaton ${\mathfrak{Prefin}}(A_{1})$ is given by
	
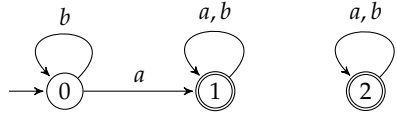
\begin{center}
\begin{tikzpicture}[node distance=2cm]
\node [state,initial] (p0) {$0$};
\node[state,accepting] (p1)[right of=p0]{$1$};
\node[state,accepting] (p2)[right of=p1]{$2$};
\path[->]
(p0)edge node{$a$} (p1)
(p0)edge[loop] node[swap] {$b$} (p0)
(p1)edge[loop ] node[swap] {$a,b$} (p1)
(p2)edge[loop] node[swap] {$a,b$} (p2)
;
\end{tikzpicture}
\end{center}
\end{example}
\begin{definition}
We consider an operation $\otimes$ acting on $k$-tuples of languages defined on the same alphabet. The operation $\otimes$ is said to be \emph{describable} (\emph{$\mathfrak{m}$-describable}) if there exists a $k$-modifier $\mathfrak{m}$ 
 such that  for any $k$-tuple of CDFA $\underline{A}$, we have $L({\mathfrak m}\underline{A})=\otimes(L(A_1),\ldots, L(A_k))$.  
\end{definition}

\begin{example}\rm
	The operation $\mathrm{Prefin}$ defined by $\mathrm{Prefin}(L)=L\Sigma^*$ for any $L\subset \Sigma^*$ is the $\mathfrak{Prefin}$-describable operation 
	where $\mathfrak{Prefin}$ is the modifier defined in Table \ref{table operations}. 
\end{example}
 For the modifiers $\mathfrak{Union}$, $\mathfrak{Inter}$ and $\mathfrak{Xor}$, a  state is an element of the  cartesian product of the states of the input. For the  $\mathfrak{Conc}$ modifier, a state is a pair composed of a state of the first input   and a subset of states of the second input. For the $\mathfrak{Star}$ modifier, a state is  a subset of states of the input.
For the $\mathfrak{SRoot}$ modifier, each state is a function from the set of states to the set of states of the input.

\begin{table}[H]\label{table-modif}
$$
\begin{array}{|c|c|c|c|c|c|}
\hline
&\mathfrak Q\underline{Q}& \iota(\underline Q,\underline i,\underline F)&\mathfrak f(\underline Q,\underline i, \underline F)& \mathfrak d(\underline{i} ,\underline F,\underline{\delta^a})\\
\hline
\mathfrak{Comp}&Q_1&i_1&Q_{1}\setminus F_{1}&\delta_1^a\\\hline
\mathfrak{Prefin}&Q_1&i_1&F_{1}&q\rightarrow \left\{\begin{array}{ll}\delta_1^a(q)&\text{ if }q\not\in F_{1}\\q&\text{ if }q\in F_{1} \end{array}\right.\\\hline
\mathfrak{Union}&Q_1\times Q_2&(i_1,i_2)&F_{1}\times Q_{2}\cup  Q_{1}\times F_{2}&{}
\underline{\delta^a}\\\hline
\mathfrak{Inter}&Q_1\times Q_2&(i_1,i_2)&F_{1}\times F_{2}&\underline{\delta^a}\\\hline
\mathfrak{Xor}&Q_1\times Q_2&(i_1,i_2)&\begin{array}{l}F_{1}\times(Q_{2}\setminus F_{2})\\
\cup (Q_{1}\setminus F_{1})\times F_{2}\end{array}&\underline{\delta^a}\\\hline
\vspace{-3mm} &&&&\\
\mathfrak{Conc}&Q_1\times2^{Q_2}&(i_1,\emptyset)& \{(q_1,E)\mid E\cap F_{2}\neq\emptyset\}&(q_{1},E)\rightarrow{}
\Xi^{F_{1}}_{i_1}(\delta_1^a(q_{1}),\delta _2^a (E))
\\
\vspace{-3mm} &&&&
\\\hline
\vspace{-3mm} &&&&\\
\mathfrak{Star}&2^{Q_1}&\emptyset&\begin{array}{l}\{E\mid E\cap F_1\neq\emptyset\}\cup\{\emptyset\}\end{array}&E\rightarrow\left\{\begin{array}{ll}\overline{\{\delta_1^a(i_1)\}}^{F_1,i_1}&\mbox{ if }E=\emptyset\\
\overline{\delta_1^a(E)}^{F_1,i_1}&\mbox{ otherwise }\end{array}\right.
\\
\vspace{-3mm} &&&&
\\\hline
\mathfrak{SRoot}&Q_1^{Q_1}&Id&\left\{g\mid g^2(i_1)\in F_{1}\right\}&g\rightarrow (\delta_1^a\circ g)
\\\hline
\multicolumn{5}{l}
{\text{where } 
\overline E^{F,x}=E\cup\{x\}\text{ if }E\cap F\neq\emptyset\text{ and }E\text{ otherwise, and }\Xi_{y}^{F}(x_,E)=(x,E\cup\{y\})\text{ if }x\in F\text{ and }(x,E)\text{ otherwise.}}
\end{array}
$$

\caption{Description of modifiers for some describable operations}\label{table operations}
\end{table}
\begin{example}[Mirror modifier]\label{Mirror-modifier}
Let us define the $1$-modifier  ${\mathfrak{Mirror}}=(\mathfrak Q,\iota, \mathfrak{f}, \mathfrak d)$ as :
\begin{itemize}
\item $\mathfrak Q(Q_1)= 2^{Q_1}$,
\item $\iota(Q_1,i_1, F_1)=F_1$,
\item $\mathfrak f(Q_1,i_1, F_1)= \{E\subset {Q_1}\mid i_1\in E\}$.
\item $\mathfrak d(i_1,F_1, \delta_1^a)$ is defined as  $E\rightarrow E' $ with   $E'=\displaystyle \bigcup _{q\in E}\{q'\mid \delta_1^a (q')= q\}$,
\end{itemize}
The mirror operation is  describable, indeed, for any DFA $A_1$, the mirror of $L(A_1)$ is  $L({\mathfrak{Mirror}}(A_1))$.
Applying  the ${\mathfrak{Mirror}}$ modifier to the  automaton $A$ of Figure~\ref{fig-aut} leads to the DFA of Figure~\ref{mirror-modifier}.
\begin{figure}[H]
\begin{center}
\begin{tikzpicture}[node distance=2cm]
			\node[state,initial] (p0) {$0$};
			\node[state,accepting] (p1) [right of=p0] {$1$};
			\node[state,accepting] (p2) [right of=p1] {$2$};
        \path[->](p0) edge node {$a$} (p1)
        (p0) edge[bend right] node {$b$} (p2)
        (p1) edge node {$a,b$} (p2)
        (p2) edge[loop] node[swap]{$a,b$} (p2);
    \end{tikzpicture}
\end{center}
\caption{The automaton $A$.}\label{fig-aut}
\end{figure}
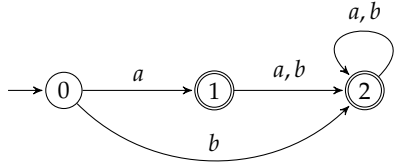
\vspace{-2cm}
\begin{figure}[H]
\begin{center}
\begin{tikzpicture}[node distance=2cm]
  \node[state] (pv) {$\emptyset$};
  \node[state,accepting] (p0) [right of= pv]{$\{0\}$};
  \node[state] (p1) [right of= p0]{$\{1\}$};
  \node[state,accepting] (p01) [below of= p0]{$\{0,1\}$};
  \node[state,accepting] (p012) [right of=p1]{$\{0,1,2\}$};
  \node[state,initial above] (p12) [right of = p012] {$\{1,2\}$};
  \node[state] (p2) [right of = p12] {$\{2\}$};
  \node[state,accepting] (p02) [below of = p12] {$\{0,2\}$};
  \path[->]
  (pv) edge[loop left] node {$a,b$} (pv)
  (p0) edge node[swap] {$a,b$} (pv)
  (p1) edge node[swap] {$a$} (p0)
  (p1) edge[bend right] node [swap]{$b$} (pv)
  (p01) edge node {$a$} (p0)
  (p01) edge node {$b$} (pv)
  (p012) edge[loop below] node {$a,b$} (p012)
  (p12) edge node {$a,b$} (p012)
  (p2) edge[bend right] node[swap] {$b$} (p012)
  (p2) edge node {$a$} (p12)
  (p02) edge node {$a$} (p12)
  (p02) edge node {$b$} (p012);

\end{tikzpicture}
\end{center}
\caption{The automaton ${\mathfrak{Mirror}}(A)$.}\label{mirror-modifier}
\end{figure}
\end{example}
For some usual operations on languages 
($\mathrm{Comp}$, $\mathrm{Union}$, $\mathrm{Inter}$, $\mathrm{Xor}$, $\mathrm{Conc}$, $\mathrm{Star}$ and $\mathrm{SRoot}$), we give  one of their modifiers in Table \ref{table operations}. Thus these  operations are describable.
\subsection{Properties}
We want to show that there exists non-describable operations (Example \ref{example-non-describable}). Corollary \ref{cor-alpha} allows us to show this fact. We also want to prove that there exists a modifier for the composition of describable operations (Section \ref{sect-appli}). In order to do this, we use the fact that the composition of two modifiers is a modifier (Proposition \ref{prop-composition} and Corollary \ref{prop-describable}).\\

We thus investigate two kinds of  properties:
\begin{itemize}
\item commutation with respect to  alphabetic renaming and restriction,
\item stability by composition.
\end{itemize}

For the first property, we consider 
    three alphabets $X,X'$ and $Y$ with $X\cap X'=\emptyset$,  a bijection $\varphi$  from $X$ to $Y$, naturally extended as an isomorphism of monoids from $X^{*}$ to   $Y^{*}$,  and $\eta:2^{(X\cup X')^{*}}\rightarrow 2^{Y^{*}}$ defined by $\eta(L)=\varphi(L\cap X^*)$. We have
\begin{claim}\label{claim1}
If $A=(X\cup X',Q,i,F,\delta)$ is a DFA recognizing a language $L$ then  $\eta (L)$ is the regular  language recognized by $(Y,Q,i,F,\delta_{\bullet})$ where $\delta_{\bullet}^y=\delta^{\varphi^{-1}(y)}$ for any $y\in Y$.
\end{claim}
 
\begin{proposition}
	Let $\otimes$ be a $k$-ary describable operation. 
		 For any $\underline{L}\in (2^{(X\cup X')^*})^k$ we have
	
	\[{}
	\otimes(\eta(L_1),\cdots,\eta(L_k))=\eta(\otimes \underline{L}).
	\]
\end{proposition}
\begin{proof}
Let $\underline A$ be a $k$-tuple
	of CDFA $A_{j}=(X\cup X', Q_{j},i_j,F_{j},\delta_j)$ such that $L(A_{j})=L_{j}$. 
	Since $\otimes$ is describable, there exists a modifier $\mathfrak m=(\mathfrak Q,\iota, \mathfrak f, \mathfrak d)$ such that $L(\mathfrak m\underline{A})=\otimes \underline{L}$.
	We have $\mathfrak m\underline{A}=(X\cup X',\mathfrak Q\underline{Q},\iota(\underline{Q},\underline{i}, \underline{F}),\mathfrak f(\underline{Q},\underline{i}, \underline{F}),\delta)$ with $ \delta^a=\mathfrak d(\underline{i},\underline{F},\underline{\delta^a})$. Then by Claim \ref{claim1}, the language $\eta(\otimes \underline{L})$ is recognized by the CDFA $A_{\square}= (Y,\mathfrak Q\underline{Q},\iota(\underline{Q},\underline{i},\underline{F}),\mathfrak f(\underline{Q},\underline{i}, \underline{F}),\delta_{\square})$ with $\delta_{\square}^a=\delta^{\varphi^{-1}(a)}=\mathfrak d(\underline{i},\underline{F},\underline{\delta^{\varphi^{-1}(a)}})$.
	Now, let $\underline{{A}_{\diamond}}$ be the $k$-tuple of CDFA $A_{{\diamond}_j}=(Y,{Q_j},i_j,F_j,\delta_{{\diamond}_j})$ with $\delta_{\diamond_j}^a=\delta_j ^{\varphi^{-1}(a)}$.  Clearly, by Claim~\ref{claim1}, $A_{{\diamond}_j}$ recognizes $\eta(L_j)$. Since $\otimes$ is describable, $\mathfrak m\underline{A_{\diamond}}=(Y,\mathfrak Q\underline{Q},\iota(\underline{Q},\underline{i},\underline{F}),\mathfrak f(\underline{Q},\underline{i},\underline{F}),\delta_{\diamond})$ with  $\delta^a_{\diamond}=\mathfrak d(\underline{i},\underline{F}, \underline{\delta_{\diamond}^a})= \mathfrak d(\underline{i}, \underline{F},\underline{\delta^{\varphi^{-1}(a)}})=\delta_\square^a$ which ends the proof.
%
\end{proof}

Immediately as special cases of the previous proposition, we obtain:
\begin{corollary}\label{cor-alpha}
	Let $\otimes$ be a $k$-ary describable  operation and $Y$ be an alphabet. Let $\underline L$ be a $k$-tuple of regular languages over $Y$. Then
	\begin{itemize}
		\item If $X\subset Y$ then $\otimes(L_{1}\cap X^{*},\cdots,L_{k}\cap X^{*})=\otimes\underline L\cap X^{*}$.
		\item For any bijection $\sigma:Y\rightarrow Y$ extended as an automorphism of monoids, 
		we have  $\otimes(\sigma(L_{1}),\cdots,\sigma(L_{k}))$ $=\sigma(\otimes\underline L)$.
	\end{itemize}
\end{corollary}
\begin{example}\label{example-non-describable}
	This result allows us to build  examples of non-describable operations.
	\begin{itemize}
		\item We consider the binary operation  defined by $\otimes(L_{1},L_{2})=L_{1}\cdot L_{2}^{-1}=\{u\mid uv\in L_{1}\mbox{ for some }v\in L_{2}\}$. This operation is not describable because it violates the first condition of Corollary \ref{cor-alpha}. For instance, let $Y=\{a,b,c\}$, $L_{1}=\{abc\}$, and $L_{2}=\{c\}$. We have
		$\otimes(L_{1}\cap\{a,b\}^{*},L_{2}\cap \{a,b\}^{*})=\emptyset.\emptyset^{-1}=\emptyset$ while 
		$\otimes(L_{1},L_{2})\cap\{a,b\}^{*}=\{ab\}$.
		\item{}We consider the unary operation  defined by $\otimes (L)=L\setminus \{a\}$ if the words $a$ and $a^{2}$ belong to $L$ and $\otimes (L)=L$ otherwise. This operation satisfies the first condition  of  Corollary \ref{cor-alpha} but it violates the second one. Indeed, if $Y=\{a,b\}$ then $\otimes (\{a,a^{2}\})=\{a^{2}\}$ while $\otimes(\{b,b^{2}\})=\{b,b^{2}\}$. So it is not describable.
	\end{itemize}
\end{example}
\begin{rmk}
	There exist $k$-modifiers that can not be associated to operations. For instance, consider the modifier $\mathfrak {Fto}1=(\mathfrak Q,\iota,\mathfrak f, \mathfrak d)$ such that
	\begin{itemize}
		\item
	$\mathfrak QQ=Q$,
	\item $\iota({Q,i,F})=i$.
	\item $\mathfrak{f}({Q,i,F})=F$, 
	\item $\mathfrak d(i,F,\delta^a_1)(q)=\delta_1^a(q)$ if $q\not\in F$ and $\mathfrak d(i,F, \delta_1^a)(q)=\left\{\begin{array}{ll}1&\text{ if }1\in Q\\\delta_1^a(q)&\text{ otherwise}\end{array}\right.$ if $q\in F$.
	
\end{itemize}
If $A_{1}$ and $A'_{1}$ are two deterministic automata recognizing the same language then we have in general $L({\mathfrak {Fto}1}(A_{1}))\neq L({\mathfrak {Fto}1}(A'_{1}))$ because the recognized language depends on the labels of the states of $A_{1}$ and $A'_{1}$. For instance,
the two following automata recognize the same language $a^{2}a^{*}$.
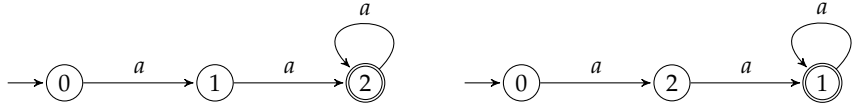
	\begin{center}
\begin{tikzpicture}[node distance=2cm]
\node [state,initial] (p0) {$0$};
\node[state] (p1)[right of=p0]{$1$};
\node[state,accepting] (p2)[right of=p1]{$2$};
\path[->]
(p0)edge node{$a$} (p1)
(p1)edge node {$a$} (p2)
(p2)edge[loop] node[swap] {$a$} (p2)
;
\end{tikzpicture}
\begin{tikzpicture}[node distance=2cm]
\node [state,initial] (p0) {$0$};
\node[state] (p2)[right of=p0]{$2$};
\node[state,accepting] (p1)[right of=p2]{$1$};
\path[->]
(p0)edge node{$a$} (p2)
(p2)edge node {$a$} (p1)
(p1)edge[loop] node[swap] {$a$} (p1)
;
\end{tikzpicture}
\end{center}
But applying ${\mathfrak {Fto}1}$ on the first one gives
	
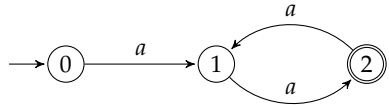
\begin{center}
\begin{tikzpicture}[node distance=2cm]
\node [state,initial] (p0) {$0$};
\node[state] (p1)[right of=p0]{$1$};
\node[state,accepting] (p2)[right of=p1]{$2$};
\path[->]
(p0) edge node{$a$} (p1)
(p1)edge[bend right] node {$a$} (p2)
(p2)edge[bend right] node[swap] {$a$} (p1)
;
\end{tikzpicture}
\end{center}
which recognizes $(aa)^{+}$ while ${\mathfrak {Fto}1}$ lets  the second automaton  unchanged.
\end{rmk}

Modifiers can be seen as functions on automata and as such can be composed as follows:
$$\mathfrak m_1\circ _j \mathfrak m_2(A_1,\ldots, A_{k_1+k_2-1})=\mathfrak m_1 (A_1,\ldots, A_{j-1},\mathfrak m_2(A_j,\ldots, A_{j+k_2-1}),A_{j+k_2},\ldots, A_{k_1+k_2-1}).$$
We have
\begin{proposition}\label{prop-composition}
The composition $\mathfrak m_1\circ _j \mathfrak m_2$ is a modifier.
\end{proposition}
\begin{proof}
Let $\mathfrak m_{1}=(\mathfrak Q^{(1)},\iota^{(1)},\mathfrak f^{(1)},\mathfrak d^{(1)})$
	be a $k_1$-ary modifier and $\mathfrak m_{2}=(\mathfrak Q^{(2)},\iota^{(2)},\mathfrak f^{(2)},\mathfrak d^{(2)})$ be a $k_2$-ary modifier. 
We define the $(k_1+k_2-1)$-ary modifier
	$(\mathfrak Q,\iota,\mathfrak f,\mathfrak d)$ by
		\begin{itemize}
		\item $\mathfrak Q\underline Q=\mathfrak Q^{(1)}\widehat{\underline Q}$ with $\widehat{\underline Q}=(Q_{1},\dots,Q_{j-1},\mathfrak Q^{(2)}(Q_{j},\ldots,Q_{j+k_{2}-1}),Q_{j+k_{2}},\dots,Q_{k_{1}+k_{2}-1})$
                \item $\iota({\underline Q,\underline i,\underline F})=\iota^{(1)}({\widehat{\underline Q},\widehat{\underline i},\widehat{\underline F}})$,\\
                  with $\widehat{\underline i}=(i_1,\ldots,i_{j-1},\iota^{(2)}((Q_j,\ldots,Q_{j+k_2-1}),(i_j,\ldots,i_{j+k_2-1}),(F_j,\ldots,F_{j+k_2-1})),i_{j+k_2},\ldots,i_{k_1+k_2-1})$ and $\widehat{\underline F}=(F_1,\ldots,F_{j-1},\mathfrak f^{(2)}((Q_j,\ldots,Q_{j+k_2-1}),(i_j,\ldots,i_{j+k_2-1}),(F_j,\ldots,F_{j+k_2-1})),F_{j+k_2},\ldots,F_{k_1+k_2-1})$.
                \item $\mathfrak f({\underline Q,\underline i,\underline F})=\mathfrak f^{(1)}({\widehat{\underline Q},\widehat{\underline i},\widehat{\underline F}})$.
		\item $\mathfrak d(\underline i,\underline F,\underline{\delta^a})=\mathfrak d^{(1)}(\widehat{\underline i},\widehat{\underline F},\widehat{\underline \delta}))$\\
                  with $\widehat{\underline \delta}=(\delta_1^a,\ldots,\delta_{j-1}^a,\mathfrak d^{(2)}((i_j,\ldots,i_{j+k_2-1}),(F_j,\ldots,F_{j+k_2-1}),(\delta_{j}^a,\ldots,\delta_{j+k_{2}-1}^a)),\delta_{j+k_2}^a,\ldots,\delta_{k_1+k_{2}-1}^a)$
	        \end{itemize}
	        We check that $(\mathfrak Q,\iota,\mathfrak f,\mathfrak d)$ acts on automata as $\mathfrak m_1\circ_j \mathfrak m_2$.
\end{proof}

 \begin{corollary}\label{prop-describable}
 	Let $\otimes$ be a $k_{1}$-ary $\mathfrak m_1$-describable operation and 
	$\oplus$ be a $k_{2}$-ary $\mathfrak m_2$-describable operation. Then the operation 
	$\otimes\circ_{j}\oplus$ 
	is a $(k_{1}+k_{2}-1)$-ary $\mathfrak (m_1\circ _j \mathfrak m_2)$-describable operation for any $j\in \{1,\ldots ,k_1\}$.
\end{corollary}
\begin{proof}
Let $L_{1},\dots,L_{k_{1}+k_{2}-1}$ be regular languages recognized respectively by  DFAs $A_1,\dots,A_{k_{1}+k_{2}-1}$.
	 Since $\oplus$ and $\otimes$ are describable, we have 
	 $$\begin{array}{ll}&\otimes\circ_j \oplus(L_1,\ldots, L_{k_{1}+k_{2}-1})\\
	 =&\otimes(L_{1},\dots,L_{j-1},
	\oplus(L_{j},\dots,L_{j+k_{2}-1}),L_{j+k_{2}},\dots,L_{k_{1}+k_{2}-1})\\=&
	L(\mathfrak m_1 (A_1,\ldots A_{j-1},\mathfrak m_2(A_j,\ldots, A_{j+k_2-1}),A_{j+k_2},\ldots, A_{k_1+k_2-1}))\\
	=&L( {\mathfrak m_{1}\circ_{j}\mathfrak m_{2}}(A_{1},\ldots,A_{k_{1}+k_{2}-1}))
	.\end{array}$$
From Proposition \ref{prop-composition}, $\mathfrak m_1\circ_j \mathfrak m_2$ is a modifier. Therefore $\otimes\circ_{j}\oplus$ is describable.
\end{proof}

\section{Monsters}\label{sect-monsters}

One-monster automata of size $n$  are minimal DFAs having $n^n$ letters representing  every function from $\IntEnt{n}$ to $\IntEnt{n}$. There are $2^n$ different $1$-monster automata depending on the set of their final states. The idea of a $k$-monster is to have a common alphabet for $k$ automata.

The idea of using combinatorial objects to denote letters has already been used by Sakoda and Sipser \cite{SS78} 
to obtain  results for two-way automata,
or by Birget \cite{Bir92} to obtain deterministic state complexity.

\subsection{Definitions}
\begin{definition}
A $k$-monster  is a $k$-tuple of automata  $\underline{\mathrm{M}}_{\underline{n},\underline{F}}=(M_1,\ldots,M_k)$ where each $M_j=(\Sigma,\IntEnt{n_j},$ $0, F_j,\delta_j)$ is defined by
\begin{itemize}
\item the common alphabet  $\Sigma=\IntEnt{n_1}^{\IntEnt{n_1}}\times \IntEnt{n_2}^{\IntEnt{n_2}}\times \cdots \times \IntEnt{n_k}^{\IntEnt{n_k}}$, 
\item the set of states  $\IntEnt{n_j}$,
\item the initial state  $0$,
\item the set of final states  $F_j$,
\item the transition function $\delta_j$  defined by $\delta_j(q,\underline{g})={g_j}(q)$ for $\underline{g}=(g_1,\ldots, g_k)\in \Sigma$, \textit{i.e.} $\underline{\delta^{\underline{g}}}=\underline{g}$.

\end{itemize}
\end{definition}


%

 

\begin{example}[$k$-monster for $k=$1 and $k=2$]
\ \\
\begin{itemize}
\item The $1$-monster  $\mathrm{M}_{2,\{1\}}$   is given by the following automaton

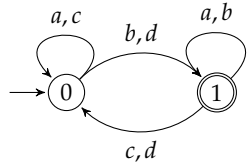
\begin{center}
\begin{tikzpicture}[node distance=2cm]
\node[state,initial](p0){$0$};{}
\node[state,accepting](p1) [right of=p0] {$1$};
\path[->]
(p0)edge[loop ] node [swap]{$a,c$} (p0)
(p0)edge[bend left] node {$b,d$} (p1)
(p1)edge[loop ] node [swap]{$a,b$} (p1)
(p1)edge[bend left] node{$c,d$}(p0);
\end{tikzpicture}
\end{center}
Each symbol codes a function $\{0,1\}\rightarrow \{0,1\}$.
\[{}
a=[01],\ b=[11],\ c=[00],\ \mbox{and } d=[10].
\]

\item The $2$-monster  $\mathrm{M}_{(2,2),(\{1\},\{1\})}$   is given by the following pair of automata on an alphabet with $2^{2}\times 2^{2}=16$ symbols where $a_{i,\_}$ (respectively $a_{\_,j}$) denotes the set of transitions $a_{i,x}$ (respectively $a_{x,j}$) for $x\in \{1,\ldots, 4\}$:
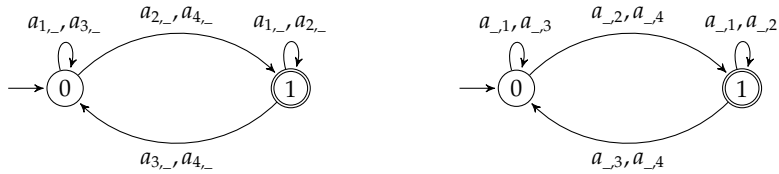
\begin{center}
			\begin{tikzpicture}[node distance=3cm]
			\node[state,initial] (p0) {$0$};
			\node[state,accepting] (p1) [right of=p0] {$1$};
			\node[state,initial] (q0) [right of=p1] {$0$};
			\node[state,accepting] (q1) [right of=q0] {$1$};
			\path[->]
			(p0) edge[loop above] node {$a_{1,\_},a_{3,\_}$} (p0)
        (p0) edge[bend left] node {$a_{2,\_},a_{4,\_}$} (p1)
        (p1) edge[loop above] node[swap] {$a_{1,\_},a_{2,\_}$} (p1)
        (p1) edge[bend left] node {$a_{3,\_},a_{4,\_}$} (p0);
        \path[->]
        (q0) edge[loop above] node {$a_{\_,1},a_{\_,3}$} (q0)
        (q0) edge[bend left] node {$a_{\_,2},a_{\_,4}$} (q1)
        (q1) edge[bend left] node {$a_{\_,3},a_{\_,4}$} (q0)
        (q1) edge[loop above] node {$a_{\_,1},a_{\_,2}$} (q1);
    \end{tikzpicture}
\end{center}

Each symbol codes a pair of functions, denoted by the word of their image. 
\[{}
\begin{array}{llll}
	a_{1,1}=[01,01]&a_{1,2}=[01,11]&a_{1,3}=[01,00]&a_{1,4}=[01,10]\\
	a_{2,1}=[11,01]&a_{2,2}=[11,11]&a_{2,3}=[11,00]&a_{2,4}=[11,10]\\
	a_{3,1}=[00,01]&a_{3,2}=[00,11]&a_{3,3}=[00,00]&a_{3,4}=[00,10]\\
	a_{4,1}=[10,01]&a_{4,2}=[10,11]&a_{4,3}=[10,00]&a_{4,4}=[10,10].
\end{array}
\]

For instance, $a_{1,2}=[01,11]$ means that the symbol $a_{1,2}$ labels a transition from $0$ to $0$ and a transition from $1$ to $1$ in the first automaton and a transition from $0$ to $1$ and a transition from $1$ to $1$ in the second automaton.\\
\end{itemize}
\end{example}

\subsection{Using monsters to compute state complexity} 
If an operation is describable, it is sufficient to study  the behavior of its modifiers over monsters to compute its state complexity.
\begin{theorem}\label{theorem-monster}
Let $\mathfrak{m}$ be a modifier and $\otimes$ be an $\mathfrak{m}$-describable operation. We have
$$\mathrm{sc}_{\otimes}\underline{n}=\mathrm{max}\{\displaystyle\#_{Min}(\mathfrak{m}\underline{\mathrm{M}}_{\underline{n},\underline{F}})\mid \underline{F}\subset \IntEnt{n_1}\times \ldots \times \IntEnt{n_k}\}.$$
\end{theorem}
\begin{proof}
	Let  $\underline{A}$ be a k-tuple of automata having   $\underline{n}$ states and having  $\underline{F}$ as set of final states recognizing a $k$-tuple of languages $\underline{L}$ over an alphabet $\Sigma$. Up to a relabelling, we assume that $A_i=(\Sigma,\IntEnt{n_i},0,F_i,\delta_i)$ for $i\in \{1,\ldots,k\}$. \\ 
         Let $\delta_A$ be the transition function of $\mathfrak{m}\underline{A}$, and $\delta_M$ the transition function of $\mathfrak{m}\underline{\mathrm{M}}_{\underline{n},\underline{F}}$. By definition of a modifier, the states of $\mathfrak{m}\underline{A}$ and of $\mathfrak{m}\underline{\mathrm{M}}_{\underline{n},\underline{F}}$ are the same. For any letter $a$, and any state $q$ of $\mathfrak{m}\underline{A}$, we have:
        \[\delta _A^a(q)=\mathfrak d((0,\ldots,0),\underline{F},\underline{\delta^a})(q)=\mathfrak d((0,\ldots,0),\underline{F},\underline{\delta^{\underline{\delta^a}}_{M}})(q)=
        \delta _M^{\underline{\delta^a}}(q).\]
        And so, for any word $w$ over alphabet $\Sigma$:
        \[\delta_A^w(q)=\delta_M^{\underline{\delta^w}}(q).\]
        Therefore, all states accessible in $\mathfrak{m}\underline{A}$ are also accessible in $\mathfrak{m}\underline{\mathrm{M}}_{\underline{n},\underline{F}}$, and, for any word $w$ over the alphabet $\Sigma$, $\delta_A^w(q) \in \mathfrak f((\IntEnt{n_1},\ldots ,\IntEnt{n_k}),(0,\ldots,0),\underline{F})$ if and only if  $\delta_M^{\underline{\delta^w}}(q) \in\mathfrak f((\IntEnt{n_1},\ldots ,\IntEnt{n_k}),$ $(0,\ldots,0),$ $\underline{F})$, which implies that all pairs of states separable in $\mathfrak{m}\underline{A}$ are also separable in $\mathfrak{m}\underline{\mathrm{M}}_{\underline{n},\underline{F}}$.
	Therefore, $$\#_{\mathrm{Min}}\mathfrak{m}\underline{A} \leq \#_{\mathrm{Min}}\mathfrak{m}\underline{\mathrm{M}}_{\underline{n},\underline{F}}.$$
\end{proof}

\begin{example}[Mirror modifier of a $1$-monster] Let us now  compute the automaton ${\mathfrak{Mirror}}(\mathrm{M}_{n_1,\{n_1-1\}})$  as in Example \ref{Mirror-modifier}.

\begin{figure}
\begin{center}
\begin{tikzpicture}[node distance=2cm, bend angle=20]
\node [state,accepting] (p0) at (0,0) {$\{0\}$};
\node[state] (pv) at (2,-1.5){$\emptyset$};
\node[state,accepting] (p01) at (-3,-1.5) {$\{0,1\}$};
\node[state,initial right] (p1)at (0,-3) {$\{1\}$};
\path[->]
(p0)edge node[swap]{$c$} (p01)
(p0)edge node{$b$} (pv)
(p0)edge[bend right] node [swap]{$d$} (p1)
(p0)edge[loop,in= 125, out=65, distance= 1cm] node [swap] {$a$} (p0)
(pv)edge[loop, in= 35, out=-25, distance= 1cm] node[swap] {$a,b,c,d$} (pv)
(p1)edge[bend right] node[swap]{$d$} (p0)
(p1)edge[loop, in= 305, out=245, distance= 1cm] node [swap]{$a$} (p1)

(p1) edge node[swap] {$b$} (p01)
(p1) edge node [swap]{$c$} (pv)
(p01) edge[loop,in= 215, out=155, distance= 1cm] node [swap]{$a,b,c,d$} (p01);
\end{tikzpicture}
\end{center}
\caption{The automaton ${\mathfrak{Mirror}}(\mathrm{M}_{2,\{1\}}$).} 
\end{figure}
We show that the automaton $\mathfrak{Mirror}(\mathrm{M}_{n_1,\{n_1-1\}})$ is minimal  when  $n_1>1$.
Indeed,
\begin{itemize}
	\item Each state is accessible. Let $g_E$ be the symbol that sends each element of a set $E\subset \IntEnt{n_1}$ to $n_1-1$ and the others ($\IntEnt{n_1}\setminus E$) to $0$. Then, we have
	$\delta^{g_E}(n_1-1)=g_E^{-1}(n_1-1)=E$ (Notice that it also works with $E=\emptyset$).
	 \item States are pairwise non-equivalent. Let $E$ and $E'$ be two distinct states of $\mathfrak{Mirror}(\mathrm{M}_{n_1,\{n_1-1\}})$. We assume there exists $i\subset E\setminus E'$. Let $g$ be the symbol sending $0$ to $i$ and the other states to $j\neq i$. The state $\delta^g(E)$ is final because $\{0\}=g^{-1}(i)\subset \delta^g(E)$ while $\delta^g(E')\subset \IntEnt{n_1}\setminus \{0\}$.
	 
\end{itemize}
\end{example}

We can describe in an algorithm the way to compute the state complexity of an operation using monsters and modifiers.
\begin{enumerate}
\item Describing the transformation with the help of a modifier whose states are represented by combinatorial objects;
\item Applying the modifier to well-chosen $k$-monsters. We will have to discuss  the final states;
\item Minimizing the  resulting automaton and estimating its size. 
\end{enumerate}
    
\section{Applications}\label{sect-appli}
\subsection{The Star of intersection example}
In this section, we illustrate our method on an operation, the star of intersection, the state complexity of which is already known \cite{SSY07}. 
After having checked the upper bound, we show that this bound is tight and that the modifier of the monster $(\mathfrak{Star}\circ \mathfrak{Inter})\mathrm{\underline{M}}$ is a witness for this operation. 

Consider the $2$-modifier $\mathfrak{Star}\circ\mathfrak{Inter}=(\mathfrak Q,\iota,\mathfrak f,\mathfrak d)$. This modifier satisfies (using Table \ref{table operations} and Proposition \ref{prop-composition})
\begin{itemize}
\item $\mathfrak Q(Q_{1},Q_{2})=2^{Q_{1}\times Q_{2}}$, 
	 \item $\iota(\underline{Q},\underline{i},\underline{F})=\emptyset$, 
  	 \item $\mathfrak f(\underline{Q},\underline{i},\underline{F})=\{E\in 2^{Q_1 \times Q_2}\mid E\cap (F_1\times F_2)\neq\emptyset \}\cup\{\emptyset\}$,
\item For $\mathfrak d(\underline{i},\underline{F},{\underline{\delta^a})}$ we have
  \[\mathfrak d(\underline{i},\underline{F}, {\underline{\delta^a})}(E)=\left\{\begin{array}{ll}\overline{\{(\delta^a_1(i_1),\delta^a_2(i_2))\}}^{F_1\times F_2,{(i_1,i_2)}}&\text{ if }E=\emptyset\\
                                                                                      \overline{(\delta^a_1,\delta^a_2)(E)}^{F_1\times F_2,{(i_1,i_2)}}&\text{ otherwise.}
                                                                                      \end{array}\right.\]
\end{itemize}
	For the following of the paper, we will assume that $Q_1=\IntEnt{n_1}$,  $Q_2=\IntEnt{n_2}$ for some $n_1, n_2 \in \N$, and $i_1=i_2=0$.
	We can see  elements of $2^{\llbracket n_{1}\rrbracket \times \llbracket n_{2}\rrbracket }$ as boolean matrices of size $n_{1} \times n_{2}$. Such a matrix will be called a tableau.
	We denote by $T_{x,y}$ the value of the tableau $T$ at row $x$ and column $y$. The number of $1s'$ in a tableau $T$ will be denoted by $\#T$.\\

	As a consequence of Corollary \ref{prop-describable}, for any pair of regular languages $(L_{1},L_{2})$ over the same alphabet and any pair of complete deterministic automata $\underline{A}=(A_{1},A_{2})$ such that $L_{1}= L(A_{1})$ and $L_{2}= L(A_{2})$ we have $(L_{1}\cap L_{2})^{*}= L(({\mathfrak{Star}\circ\mathfrak{Inter}})\underline{A})$.



           Now, let $n_{1}$ and $n_{2}$ be two positive integers and let $ (F_{1},F_{2})$ be a subset of $ \llbracket n_{1}\rrbracket\times \llbracket n_{2}\rrbracket$. 
           An upper bound of the  state complexity of the composition of \emph{star} and \emph{inter} operations is obtained by maximizing the number of  states of $\mathrm{\widehat{M}}_{F_1,F_2}$ where $\mathrm{\widehat{M}}_{F_1,F_2}$ is the automaton deduced from $(\mathfrak{Star}\circ\mathfrak{Inter})\mathrm{\underline{M}}_{\underline{n},(F_1,F_2)}$ by removing tableaux having a $1$ in $(x,y)\in F_1\times F_2$  but no $1$ in $(0,0)$. Indeed, from Table \ref{table operations} the initial state of $\mathfrak{Inter}\mathrm{\underline{M}}_{\underline{n},(F_1,F_2)}$ is $(0,0)$ and its set of final states is $F_1\times F_2$. Furthermore, again from Table \ref{table operations}, if $A=(\Sigma,Q,q_0,F,\delta)$, the accessible states of $\mathfrak{Star}(A)$ are subsets $E\subset Q$ satisfying $E\cap F\neq \emptyset \Rightarrow q_0\in E$. Combining these two facts, the accessible states of  $(\mathfrak{Star}\circ\mathfrak{Inter})\mathrm{\underline{M}}_{\underline{n},(F_1,F_2)}$ are subsets $E\subset \IntEnt{n_1}\times \IntEnt{n_2}$ satisfying $E\cap (F_1\times F_2)\neq \emptyset \Rightarrow (0,0)\in E$.

  
              We first remark that  the initial state of $\mathfrak{Inter}\mathrm{\underline{M}}_{\underline{n},(0,0)}$ is the only final state. This implies that $L((\mathfrak{Star}\circ\mathfrak{Inter})\mathrm{\underline{M}}_{\underline{n},(0,0)})=L(\mathfrak{Inter}\mathrm{\underline{M}}_{\underline{n},(0,0)})^*=L(\mathfrak{Inter}\mathrm{\underline{M}}_{\underline{n},(0,0)})$, which in turn implies that $ \#_{\mathrm{Min}}(\mathrm{\widehat{M}}_{0,0})\leq \#_{\mathrm{Min}}(\mathfrak{Inter}\mathrm{\underline{M}}_{\underline{n},(0,0)}) \leq n_{1}n_{2}$.\\
          Notice also that if $\#(F_1\times F_2)=0$, then $\mathrm{\widehat{M}}_{F_1,F_2}$ recognizes the empty language, which trivially implies that $\#_{\mathrm{Min}}(\mathrm{M})\leq 1$.

\begin{lemma}\label{maximize}
The maximal number of states of $\mathrm{\widehat{M}_{F_1,F_2}}$ with $F_1\times F_2\not\in \{\{(0,0)\},\emptyset\}$ is when $\#(F_1\times F_2)=1$.
\end{lemma}
	\begin{proof}

        $\begin{array}{llll}
          \#\mathrm{\widehat{M}}_{F_1,F_2}&=&\#2^{\llbracket n_{1}\rrbracket \times \llbracket n_{2}\rrbracket}-\#\{T\in 2^{\llbracket n_{1}\rrbracket \times \llbracket n_{2}\rrbracket } \mid (\exists (x,y)\in F_{1}\times F_{2} \mbox{ s.t. } T_{x,y}=1) \land T_{0,0}=0\}\\
          &=&2^{n_1n_2}-(\#\{T\in 2^{\llbracket n_{1}\rrbracket \times \llbracket n_{2}\rrbracket } \mid T_{0,0}=0\}\\
          &&\ \ \ \ \ \ \ \ \ \ \ \ \ -\ \#\{T\in 2^{\llbracket n_{1}\rrbracket \times \llbracket n_{2}\rrbracket }\mid \forall (x,y)\in F_{1}\times F_{2}, T_{x,y}=0 \land T_{0,0}=0\})\\
          &=&\left\{\begin{array}{ll}2^{n_1n_2}-(2^{n_1n_2-1}-2^{n_1n_2-\#F_1\#F_2-1})&\text{ if }(0,0)\not\in F_1\times F_2\\
          2^{n_1n_2}-(2^{n_1n_2-1}-2^{n_1n_2-\#F_1\#F_2})&\text{ otherwise}
          \end{array}\right.
         \end{array}$
 
        In conclusion, the maximal number of states of  $\#\mathrm{\widehat{M}}_{F_1,F_2}$ with $F_1\times F_2\not\in \{\{(0,0)\},\emptyset\}$ is reached when $\#F_1\times \#F_2=1$ and is 
     $\frac{3}{4}2^{n_1n_2}$.        
        \end{proof}
        
	\begin{corollary}\label{corollary-upperbound}
            $\#_{\mathrm{Min}}((\mathfrak{Star}\circ\mathfrak{Inter})\mathrm{\underline{M}}_{\underline{n},(F_1,F_2)})\leq \frac{3}{4}2^{n_{1}n_{2}} $
          
        \end{corollary}
        \begin{proof}
        From Lemma \ref{maximize}, we maximize the number of tableaux when $\#F_1\times \#F_2=1$. So the upper bound is  $2^{n_1n_2}-(2^{n_1n_2-1}-2^{n_1n_2-1-1})= \frac{3}{4}2^{n_1n_2}$.
        \end{proof}
        
        Now we show that this upper bound is the state complexity of  the combination of the star and the intersection operations.
        
        Let $F_1,F_2$ be $\{n_1-1\},\{n_2-1\}$ and let  $\mathrm{\widehat{M}}=\mathrm{\widehat{M}}_{F_1,F_2}$.

        \begin{lemma}\label{proof-lemma-distinct-Star-inter}
          All states of $\mathrm{\widehat{M}}$ are accessible.
        \end{lemma}
        \begin{proof}
%
          Let $T$ be a state of $\mathrm{\widehat{M}}$.
          Let us define an order $<$ on tableaux as
          $T < T'$ if and only if\\
\smallskip
      \hspace*{0.5cm}   $\mathbf{(1)}$ $\#(T) < \#(T')$ or\\
    \smallskip
      \hspace*{0.5cm}   $\mathbf{(2)}$ $(\#(T) = \#(T')$ and $T_{n_1-1,n_2-1}=1$ and $T'_{n_1-1,n_2-1}=0)$ or\\
       \smallskip
       \hspace*{0.5cm}   $\mathbf{(3)}$ $(\#(T) = \#(T')$ and $T_{n_1-1,n_2-1}=T'_{n_1-1,n_2-1}$ and $T_{0,0}=1$ and $T'_{0,0}=0)$.
          
          Let us prove the assertion by induction on non-empty tableaux of $\mathrm{\widehat{M}}$ for the partial order $<$ (the empty tableau is the initial state of $\mathrm{\widehat{M}}$, and so it is accessible):\\
          The only minimal tableau for non-empty tableaux of $\mathrm{\widehat{M}}$ and the order $<$ is the tableau with only one $1$ at $(0,0)$. This is accessible from the initial state $\emptyset$ by reading the letter $(\mathrm{Id},\mathrm{Id})$. Let us notice that each letter is a pair of functions of $\IntEnt{n_1}^{\IntEnt{n_1}}\times \IntEnt{n_2}^{\IntEnt{n_2}}$.\\
          Now let us take a tableau $T'$, and find a tableau $T$ such that $T<T'$, and T' is accessible from $T$. We distinguish the cases below, according to some properties of $T'$. For each case, we define a tableau $T$ and a letter $(f,g)$. For all cases, except the last one, we easily check that \\
          \smallskip
      \hspace*{1cm}
          \textbf{(1)} $T_{0,0}=1$ (which implies that $T$ is a state of $\mathrm{\widehat{M}}$),\\
          \smallskip
      \hspace*{1cm}
           \textbf{(2)} $\delta^{(f,g)}(T)=(f,g)(T)=T'$ (where $(f,g)(T)=\{(f(i),g(j))\mid (i,j)\in T\}$), and\\
           \smallskip
      \hspace*{1cm}
 \textbf{(3)} $T < T'$.

          \begin{itemize}
          \item $T'_{n_1-1,n_2-1}=0$.
          \smallskip
            \begin{itemize}
          \item $T'_{0,0}=0$. Let $(i,j)$ be the index of a $1$ in $T'$. Define $(f,g)$ as $((0,i),(0,j))$ where $(0,i)$ and $(0,j)$ denote transpositions, and $T=(f,g)(T')$.
          \item $T'_{0,0}=1$.
            \begin{itemize}
            \item There exists $(i,j) \in \{1,2,...,n_{1}-1\} \times \{1,2,..., n_{2}-1\}$ such that $T'_{i,j}=1$.
             Define $(f,g)$ as $((n_1-1,i),(n_2-1,j))$, then $T=(f,g)(T')$.
            \item For all $(i,j) \in \{1,2,...,n_{1}-1\} \times \{1,2,..., n_{2}-1\}$, $T'_{i,j}=0$, $T'_{0,n_2-1}=1$ and $T'_{n_1-1,0}=1$.
              In that case, define $(f,g)$ as $(Id,(n_2-1,0))$, and $T$ as $(f,g)(T')$.
            \item For all $(i,j) \in \{1,2,...,n_{1}-1\} \times \{1,2,..., n_{2}-1\}$, $T'_{i,j}=0$, $T'_{0,n_2-1}=1$ and $T'_{n_1-1,0}=0$.
             Define $(f,g)$ as $\left(\left(n_1-1\atop 0\right),Id\right)$ .  Then $T$ is defined as
              \[ \left\{\begin{array}{lll} T_{0,n_2-1}=0 \\
              T_{n_1-1,n_2-1}=1\\
              T_{i,j}=T'_{i,j}\mbox{ if }(i,j)\notin {(0,n_2-1),(n_1-1,n_2-1)}
              \end{array}\right. \]
            \item For all $(i,j) \in \{1,2,...,n_{1}-1\} \times \{1,2,..., n_{2}-1\}$, $T'_{i,j}=0$, $T'_{0,n_2-1}=0$ and $T'_{n_1-1,0}=1$.
              This case is symmetrical to the case above.
            \item For all $(i,j) \in \{1,2,...,n_{1}-1\} \times \{1,2,..., n_{2}-1\}$, $T'_{i,j}=0$, $T'_{0,n_2-1}=0$ and $T'_{n_1-1,0}=0$.
             Let $(i,j)\neq (0,0)$ be a $1$ in $T'$. Define $(f,g)=\left(\left(n_1-1\atop i\right), \left(n_2-1\atop j\right)\right)$, and define $T$ as follows 
              \[ \left\{\begin{array}{lll} T_{i,j}=0 \\
              T_{n_1-1,n_2-1}=1\\
              T_{i',j'}=T'_{i',j'}\mbox{ if }(i',j')\notin {(i,j),(n_1-1,n_2-1)}
              \end{array}\right. \]
            \end{itemize}
            \end{itemize}
          \item $T'_{0,0}=1$ and $T'_{n_1-1,n_2-1}=1$.
            Let $(f,g)=((n_1-1,0),(n_2-1,0))$. Let $T''$ be the matrix obtained from $T'$ by replacing the $1$ in $(0,0)$ by a $0$. Let $T= (f,g)(T'')$. As $(f,g)$ is a bijection over $\llbracket n_{1}\rrbracket \times \llbracket n_{2}\rrbracket$, we have $T_{0,0}=((f,g)(T''))_{0,0}=T''_{n_1-1,n_2-1}=1$, which means that $T$ is a state of $\mathrm{\widehat{M}}$, and $(f,g)(T)=(f,g)(f,g)(T'')=T''$. As $T''_{n_1-1,n_2-1}=1$, we have $\delta^{(f,g)}(T)= T' $ in $\mathrm{\widehat{M}}$. Furthermore, $\#T < \#T'$ implies that $T < T'$.
            \end{itemize}
        \end{proof}

        
        \begin{lemma}\label{proof-lemma-dist-Star-inter}
          All states of $\mathrm{\widehat{M}}$ are separable.
        \end{lemma}
        \begin{proof}
          Let $T$ and $T'$ be two different states of $\mathrm{\widehat{M}}$. There exists $(i,j)\in \IntEnt{n_1}\times\IntEnt{n_2}$ such that $T_{i,j}\neq T'_{i,j}$. Suppose, for example, that $T_{i,j}=1$ and $T'_{i,j}=0$. Let $(f,g)\in \IntEnt{n_1}^{\IntEnt{n_1}}\times \IntEnt{n_2}^{\IntEnt{n_2}}$ such that~:
          \[f(x)=\left\{\begin{array}{ll} n_1-1 &\text{ if } x=i \\
          0 &\text{ otherwise}
          \end{array}\right.
          \mbox{ and }
          g(x)=\left\{\begin{array}{ll} n_2-1 &\text{ if } x=j \\
          0 &\text{ otherwise}
          \end{array}\right.\]
          We have  $ \delta^{(f,g)}(T)_{n_1-1,n_2-1}=T_{i,j}=1$, and $ \delta^{(f,g)}(T')_{n_1-1,n_2-1}=T'_{i,j}=0$. Therefore, $T$ and $T'$ are separable in $\mathrm{\widehat{M}}$.         
        \end{proof}
        
        Theorem \ref{theorem-monster} and the previous lemmas give us :
        \begin{theorem}\label{theo-Star-inter}
          The state complexity of the star of intersection is $\frac{3}{4}2^{n_1n_2}$.
        \end{theorem}

\subsection{The square root example}
 In this section, we are interested in the square root of a language $L$  defined by $\sqrt L=\{x\mid xx\in L\}$. By a straightforward computation, we show easily that $sc_{\sqrt\ }(2)=2$ and thus we investigate only  the case $n>2$. Maslov \cite{Mas70} showed that the square root preserves regularity and he gave a construction that can be summarized in terms of modifier by ${\mathfrak{SRoot}}$ (see Table \ref{table operations}).
We first remark that this construction gives us an upper bound of $n^{n}$ for the state complexity of square root. We also notice that in this case  $\mathfrak d$ is a morphism in the sense  that $\mathfrak d(0,F,\delta^a)\circ \mathfrak d(0,F,\delta^b)=\mathfrak d(0,F,\delta^a\circ \delta^b)$. Therefore, the application $\phi \rightarrow \mathfrak d(0,F,\phi)$ is a morphism of semigroups from $\llbracket n\rrbracket^{\llbracket n\rrbracket}$ to $\mathfrak d(0,F,\llbracket n\rrbracket^{\llbracket n\rrbracket})$. As $\llbracket n\rrbracket^{\llbracket n\rrbracket}$ can be generated by 3 elements, $\mathfrak d(0,F,\llbracket n\rrbracket^{\llbracket n\rrbracket})$ can be too. Therefore, there exists a witness with at most $3$ letters. These two properties have been already noticed by Maslov in \cite{Mas70}.

Let us consider the automaton ${\mathfrak{SRoot}}(\mathrm M_{n,F})$. 
We notice that all the states in ${\mathfrak{SRoot}}(\mathrm M_{n,F})$ are accessible. Indeed, the state labeled by the function $g$ is reached from $Id$ by reading the letter $g$.

For the separability, we consider a state $g_{a,b}$ defined as follows.
Let  $a\neq b\in\llbracket n\rrbracket$ and 
$g_{a,b}(x)=a$ if $x\in F$ and $g_{a,b}(x)=b$ otherwise. 

\begin{lemma}
For each pair $a,b\in \llbracket n\rrbracket$ such that  $a\neq b$, the two states $g_{a,b}$ and $g_{b,a}$ are not separable in ${\mathfrak{SRoot}}(\mathrm M_{n,F})$.
\end{lemma}
\begin{proof}
Let us prove that for any $h$, the functions
$h\circ g_{a,b}$ and $h\circ g_{b,a}$ are both final or both non final.
In fact we have only two values of $h$ to investigate: $h(a)$ and $h(b)$. If $h(a),h(b)\in F$ or  $h(a),h(b)\not\in F$ then the two functions $h\circ g_{a,b}$ and $h\circ g_{b,a}$ are obviously both final or both non final.
Without loss of generality, suppose that $h(a)\in F$ (and so $h(b)\not\in F$). We have to examine two possibilities:
\begin{itemize}
\item Either $0\in F$, in this case $h(g_{a,b}(0))=h(a)\in F$. Then $g_{a,b}(h(g_{a,b}(0)))=a$ and $h(g_{a,b}(h(g_{a,b}(0))))=h(a)\in F$.
But $h(g_{b,a}(0))=h(b)\not\in F$. Hence,  $g_{b,a}(h(g_{b,a}(0)))=a$, so $h(g_{b,a}(h(g_{b,a}(0))))\in F$. This implies that the two states are final.
\item Or $0\not\in F$, in this case $h(g_{a,b}(0))=h(b)\not\in F$. Then  $g_{a,b}(h(g_{a,b}(0)))=b$ and $h(g_{a,b}(h(g_{a,b}(0))))=h(b)\not\in F$. But we also have $h(g_{b,a}(0))=h(a)\in F$. Hence,  $g_{b,a}(h(g_{b,a}(0)))=b$, so $h(g_{b,a}(h(g_{b,a}(0))))$ $\not\in F$. This implies that the two states are not final.
\end{itemize}
We deduce that the two states are not separable. Notice that the number of transformations $g_{a,b}$ is exactly $2\binom{n}{2}$.
\end{proof}

\begin{corollary}\label{cor-sqrt}
 $$\mathrm{sc}_{\sqrt\ }(n)\leq n^n-\binom n2$$
 \end{corollary}
 Notice that the state complexity is lower than the bound given by Maslov \cite{Mas70}.\\
 
 \begin{lemma}\label{lm-g}
Let $F=\{n-1\}$, and   $P=\{(g,g')\mid g\neq g'\mbox{ and }\forall a,b\in\llbracket n\rrbracket, (g,g')\neq (g_{a,b},g_{b,a}) \}$. 
For any pair of distinct states $(g,g')\in P$, $g$ and $g'$ are separable  in ${\mathfrak{SRoot}}(\mathrm M_{n,F})$. 
\end{lemma}
\begin{proof}
Three cases have to be considered:
\begin{itemize}
\item \textbf{Suppose  that $\mathbf{g(0)=g'(0)}$.}\\ Then there exists $x\in\llbracket n\rrbracket\setminus\{0\}$ such that $g(x)\neq g'(x)$. We set $h(g(0))=x$. Hence, $h(g(h(g(0))=h(g(x))$ and $h(g'(h(g'(0))=h(g'(x))$.  But, as $g(x)\neq g'(x)$, it is always possible to choose $h$ such that $h(g(x))=n-1$ while $h(g' (x))\neq n-1$. Thus $h\circ g$ is a final state while $h\circ g'$ is not.
\item \textbf{Suppose that $\mathbf{g(0)\neq g'(0)}$ and that $\mathbf{\#(Im(g)\cup Im(g'))>2}$.}\\ Without loss of generality, one assumes that there exists $x\in \mathrm{Im}(g)$ such that $x\not\in\{g(0),g'(0)\}$. So the values $h(g(0))$, $h(g'(0))$ and $h(x)$ can be chosen independently each from the others. We set $h(g(0))=y$ with $g(y)=x$, $h(g'(0))=0$ and $h(x)=n-1$. We check that $h\circ g$ is a final state while $h\circ g'$ is not final.
\item \textbf{Suppose that $\mathbf{g(0)\neq g'(0)}$ and that $\mathbf{\#(Im(g)\cup Im(g'))=2}$.} \\
If we suppose that for any non final state $x$, we have $g(x)\neq g(n-1)$ and $g'(x)\neq g'(n-1)$. Since $x$ is not final and $\#(Im(g)\cup Im(g'))=2$, we have $g(x)=g(0)$ and $g'(x)=g'(0)$ (recall that $0$ is not final). So, as $g(0)\neq g'(0)$, this implies $g(n-1)\neq g'(n-1)$. 
In other words, $g=g_{a,b}$ and $g'=g_{b,a}$ for some $a,b$. By contraposition, if $(g,g')\neq (g_{a,b},g_{b,a})$ for any $a,b$ then there exists $x\neq n-1$ such that $g(x)=g(n-1)$ or $g'(x)=g'(n-1)$. 
Let us denote by $m$ the minimal element of $\llbracket n \rrbracket$ having this property and without loss of generality assume that $g(m)=g(n-1)$ (in particular, it means that for any $p<m$, $g'(p)\neq g'(n-1)$).
We have  two cases to consider. If $m=0$ then we set $h(g(0))=n-1$ and $h(g'(0))=0$. Obviously, $h(g'(h(g'(0))))=0$. On the other hand, $h(g(h(g(0)))=h(g(n-1))=h(g(0))=n-1$. Hence, $h\circ g$ is final while $h\circ g'$ is not final. If $m>0$ then we have $g(m)=g'(0)$ (because there are exactly two values in the image of $g$ and $g'$). Furthermore, 
 $g'(n-1)\neq g'(0)$ and so $g'(n-1)=g(0)$.  We set $h(g(0))=m$ and $h(g'(0))=n-1$. We have $h(g(h(g(0))))=h(g(m))=h(g'(0))=n-1$. On the other hand, $h(g'(h(g'(0))))=h(g'(n-1))=h(g(0))=m\neq n-1$. It follows that $h\circ g$ is final while $h\circ g'$ is not final.
\end{itemize}
\end{proof}

The following theorem follows directly from  Corollary \ref{cor-sqrt} and Lemma \ref{lm-g}. 
\begin{theorem}
	$\mathrm{sc}_{\sqrt\ }(n)=n^{n}-\binom n2.$
\end{theorem}
 In order to show some restriction for alphabets of size $<3$, some lemmas will be given. We assume that $n>2$.
 	\begin{lemma}\label{lsubmonT}
		Any submonoid   of  $\IntEnt{n}^{\IntEnt{n}}$ generated by two distinct elements is a proper submonoid of $\IntEnt{n}^{\IntEnt{n}}$. 
	\end{lemma}
	\begin{proof} Suppose that $\llbracket n\rrbracket^{\llbracket n\rrbracket}$ is generated by two elements $f$ and $g$.
		Recall first that we need at least two permutations for generating the symmetric group $\mathfrak S_{n}$ (see \cite{GM08}).  We notice also that
		$\llbracket n\rrbracket^{\llbracket n\rrbracket}\setminus \mathfrak S_{n}$  is an ideal of  $\llbracket n\rrbracket^{\llbracket n\rrbracket}$, that is if $t\in \llbracket n\rrbracket^{\llbracket n\rrbracket}\setminus \mathfrak S_{n}$  and $t'\in \llbracket n\rrbracket^{\llbracket n\rrbracket}$ then $t\circ t', t'\circ t\in \llbracket n\rrbracket^{\llbracket n\rrbracket}\setminus \mathfrak S_{n}$. This shows that $f, g\in\mathfrak S_{n}$. But since $\mathfrak S_{n}$ is a submonoid, it is stable by composition. It follows that $\mathfrak S_{n}=\llbracket n\rrbracket^{\llbracket n\rrbracket}$. Since this is absurd we deduce the result.
	\end{proof}
	\begin{lemma}\label{lgenT}
	The monoid $\IntEnt{n}^{\IntEnt{n}}$ is generated by the two permutations $(0,1)$ and $(0,1,\dots, n-1)$ together with any of the contractions $\left(i\atop j\right)$. 
	\end{lemma}
	\begin{proof}
		It is known (see \cite{GM08}) that $\IntEnt{n}^{\IntEnt{n}}$ is generated by $(0,1)$, $(0,1,\dots, n-1)$ and $\left(0\atop 1\right)$ and that any permutation is generated by $(0,1)$ and $(0,1,\dots, n-1)$. Thus, let $\sigma$ be a permutation sending $0$ to $i$ and $1$ to $j$. The result is just a consequence of the equality $\sigma^{-1}\circ \left(i\atop j\right)\circ \sigma=\left(0\atop 1\right)$.
	\end{proof}

	\begin{lemma}\label{lm-limite}
		There exists $i,j\in\llbracket n\rrbracket$ such that $i\neq j$ and  $\left(i\atop j\right)\not\in \{g_{p,q}\mid p,q\in\llbracket n\rrbracket, p\neq q\}$.
	\end{lemma}
	\begin{proof}
		Since $n>2$, either $\#F>1$ or $n-\#F>1$. Assume without loss of generality $\#F>1$. Let $i,i'\in F$ and $j\in \IntEnt{n}\setminus F$.  		
		We have  $\left( i\atop j\right)\not
		\in \{g_{p,q}\mid p,q\in\llbracket n\rrbracket, p\neq q\}$ otherwise we must have $g_{p,q}(i)=p=j$ and  $g_{p,q}(i')=p=i'$ which is impossible since $j\neq i'$.
	\end{proof}
	
 \begin{proposition}\label{2-letters}
	 For any regular language $L$ over an alphabet with at most two letters, if $\mathrm{sc}(L)=n>2$ then
	$\mathrm{sc}(\sqrt L)<\mathrm{sc}_{\sqrt\ }(n)$.
 \end{proposition}

\begin{proof}
Let $L$ be a language with $\mathrm{sc}(L)=n>2$ and $A=(\{\mathtt a,\mathtt b\},\llbracket n\rrbracket, \{0\},F,\cdot)$ 
be a minimal CDFA recognizing $L$. Since $\mathrm{sc}(L)=n>2$, the set of final states $F$ is a proper subset of $\llbracket n\rrbracket$, otherwise $L=\Sigma^*$ and $\mathrm{sc}(L)=1$. 
	Since the application $\phi \rightarrow \mathfrak d(0,F,\phi)$ is a morphism of semigroups, the set of the states of $\mathfrak{SRoot}(A)$ is a submonoid $M$ of $\llbracket n\rrbracket^{\llbracket n\rrbracket}$ generated by two elements. Now, we just have to prove that any submonoid of $\IntEnt{n}^{\IntEnt{n}}$ generated by two elements has strictly less than $n^n - \binom{n}{2}$ elements.  
	 Suppose that 
 $\#_{\mathrm{Min}}(\mathfrak{SRoot}(A))=\mathrm{sc}_{\sqrt{\ } }(n)$. Thus we have $t\not\in\{g_{p,q}\mid p,q\in\llbracket n\rrbracket, p\neq q\}$ implies $t\in M$.
 	  Obviously, we have
	$(0,1), (0,1,\dots, n-1)\not\in \{g_{p,q}\mid p,q\in\llbracket n\rrbracket, p\neq q\}$ and so $(0,1), (0,1,\dots, n-1)\in M$. Furthermore, we have:
	From Lemma \ref{lm-limite} there exists $i,j\in\llbracket n\rrbracket$ such that $i\neq j$ and $\left(i\atop j\right)\in M$. So by Lemma \ref{lgenT}, $M=\llbracket n\rrbracket^{\IntEnt{n}}$. But, by Lemma \ref{lsubmonT}, as $M$ has only two generators, it is a proper submonoid of $\llbracket n\rrbracket^{\IntEnt{n}}$ which contradicts the previous sentence.
	So there exists a transformation $t\not\in \{g_{p,q}\mid p,q\in\llbracket n\rrbracket, p\neq q\}$ such that $t\not\in M$ and thus 
	$\#_{\mathrm{Min}}(\mathfrak{SRoot}(A))<n^n-\binom{n}{2}=\mathrm{sc}_{\sqrt{\ } }(n)$.
	\end{proof}
	
	Let us notice that
	Krawetz \textit{et al.} \cite{KLS05} found a very similar result 
not quite for square root, but for the closely related operation  $\mathrm{Root}(L)=\{w\mid \exists n$ such that $w^n\in L\}$.

	\section{Conclusion}
	New tools for computing state complexity are provided. As there is a witness among monster automata, one can focus on them to obtain a tight bound for state complexity. One of our future works is to use these tools on operations where the bound is not tight or not known as cyclic shift or  star of xor.	
	As these tools produce very large size alphabet, it remains to study how it is possible to improve this size by obtaining in some cases a constant size alphabet. 
	
	The authors learned that Sylvie Davies has independently and in the same time obtained  some of the results presented in this paper. In particular, she described our monster approach as the OLPA (One letter Per Action) one. Her work can be found in \cite{Dav18}.

 \bibliography{../COMMONTOOLS/biblio,../COMMONTOOLS/bibjg}

\begin{thebibliography}{10}

\bibitem{Bir92}
Jean-Camille Birget.
\newblock Intersection and union of regular languages and state complexity.
\newblock {\em Inf. Process. Lett.}, 43(4):185--190, 1992.

\bibitem{Brz13}
Janusz~A. Brzozowski.
\newblock In search of most complex regular languages.
\newblock {\em Intern. J. of Foundations of Comp. Sc.}, 24(6):691--708, 2013.

\bibitem{CLMP15}
Pascal Caron, Jean{-}Gabriel Luque, Ludovic Mignot, and Bruno Patrou.
\newblock State complexity of catenation combined with a boolean operation: {A}
  unified approach.
\newblock {\em Int. J. Found. Comput. Sci.}, 27(6):675--704, 2016.

\bibitem{CLP17}
Pascal Caron, Jean{-}Gabriel Luque, and Bruno Patrou.
\newblock State complexity of catenation combined with boolean operations.
\newblock {\em CoRR}, abs/1707.03174, 2017.

\bibitem{CLP18}
Pascal Caron, Jean{-}Gabriel Luque, and Bruno Patrou.
\newblock State complexity of multiple catenations.
\newblock {\em Fundam. Inform.}, 160(3):255--279, 2018.

\bibitem{CGKY11}
Bo~Cui, Yuan Gao, Lila Kari, and Sheng Yu.
\newblock State complexity of two combined operations: Catenation-union and
  catenation-intersection.
\newblock {\em Int. J. Found. Comput. Sci.}, 22(8):1797--1812, 2011.

\bibitem{Dav18}
Sylvie Davies.
\newblock A general approach to state complexity of operations: Formalization
  and limitations.
\newblock In Mizuho Hoshi and Shinnosuke Seki, editors, {\em Developments in
  Language Theory - 22nd International Conference, {DLT} 2018, Tokyo, Japan,
  September 10-14, 2018, Proceedings}, volume 11088 of {\em Lecture Notes in
  Computer Science}, pages 256--268. Springer, 2018.

\bibitem{Dom02}
Michael Domaratzki.
\newblock State complexity of proportional removals.
\newblock {\em Journal of Automata, Languages and Combinatorics},
  7(4):455--468, 2002.

\bibitem{GM08}
Olexandr Ganyushkin and Volodymyr Mazorchuk.
\newblock {\em {Classical finite transformation semigroups: an introduction}}.
\newblock Algebra and Applications. Springer, Dordrecht, 2008.

\bibitem{GMRY17}
Yuan Gao, Nelma Moreira, Rog{\'{e}}rio Reis, and Sheng Yu.
\newblock A survey on operational state complexity.
\newblock {\em Journal of Automata, Languages and Combinatorics},
  21(4):251--310, 2017.

\bibitem{GSY08}
Yuan Gao, Kai Salomaa, and Sheng Yu.
\newblock The state complexity of two combined operations: Star of catenation
  and star of reversal.
\newblock {\em Fundam. Inform.}, 83(1-2):75--89, 2008.

\bibitem{HU79}
John~E. Hopcroft and Jeffrey~D. Ullman.
\newblock {\em Introduction to Automata Theory, Languages and Computation}.
\newblock Addison-Wesley, Reading, MA, 1979.

\bibitem{JJS05}
Jozef Jir{\'a}sek, Galina Jir{\'a}skov{\'a}, and Alexander Szabari.
\newblock State complexity of concatenation and complementation.
\newblock {\em Int. J. Found. Comput. Sci.}, 16(3):511--529, 2005.

\bibitem{Jir05}
Galina Jir{\'a}skov{\'a}.
\newblock State complexity of some operations on binary regular languages.
\newblock {\em Theor. Comput. Sci.}, 330(2):287--298, 2005.

\bibitem{JO08}
Galina Jir{\'{a}}skov{\'{a}} and Alexander Okhotin.
\newblock State complexity of cyclic shift.
\newblock {\em {ITA}}, 42(2):335--360, 2008.

\bibitem{JO11}
Galina Jir{\'a}skov{\'a} and Alexander Okhotin.
\newblock On the state complexity of star of union and star of intersection.
\newblock {\em Fundam. Inform.}, 109(2):161--178, 2011.

\bibitem{KLS05}
Bryan Krawetz, John Lawrence, and Jeffrey Shallit.
\newblock State complexity and the monoid of transformations of a finite set.
\newblock {\em Int. J. Found. Comput. Sci.}, 16(3):547--563, 2005.

\bibitem{Mas70}
A.~N. Maslov.
\newblock Estimates of the number of states of finite automata.
\newblock {\em Soviet Math. Dokl.}, 11:1373--1375, 1970.

\bibitem{SS78}
William~J. Sakoda and Michael Sipser.
\newblock Nondeterminism and the size of two way finite automata.
\newblock In {\em Proceedings of the Tenth Annual ACM Symposium on Theory of
  Computing}, STOC '78, pages 275--286, New York, NY, USA, 1978. ACM.

\bibitem{SSY07}
Arto Salomaa, Kai Salomaa, and Sheng Yu.
\newblock State complexity of combined operations.
\newblock {\em Theor. Comput. Sci.}, 383(2-3):140--152, 2007.

\bibitem{Yu01a}
Sheng Yu.
\newblock State complexity of regular languages.
\newblock {\em Journal of Automata, Languages and Combinatorics}, 6(2):221,
  2001.

\bibitem{YZS94}
Sheng Yu, Qingyu Zhuang, and Kai Salomaa.
\newblock The state complexities of some basic operations on regular languages.
\newblock {\em Theoret. Comput. Sci.}, 125(2):315--328, 1994.

\end{thebibliography}




\end{document}